\documentclass[american,aps,pra,nofootinbib,longbibliography,onecolumn]{revtex4-2}
\usepackage[T1]{fontenc}
\usepackage[latin9]{inputenc}
\usepackage{color}
\usepackage{babel}
\usepackage{amsmath}
\usepackage{amsthm}
\usepackage{amssymb}
\usepackage{graphicx}
\usepackage[pdfusetitle,
 bookmarks=true,bookmarksnumbered=false,bookmarksopen=false,
 breaklinks=false,pdfborder={0 0 0},pdfborderstyle={},backref=false,colorlinks=true]
 {hyperref}
\hypersetup{
 allcolors=magenta}

\makeatletter
\theoremstyle{plain}
\newtheorem{thm}{\protect\theoremname}
\theoremstyle{plain}
\newtheorem{prop}[thm]{\protect\propositionname}
\ifx\proof\undefined
\newenvironment{proof}[1][\protect\proofname]{\par
	\normalfont\topsep6\p@\@plus6\p@\relax
	\trivlist
	\itemindent\parindent
	\item[\hskip\labelsep\scshape #1]\ignorespaces
}{%
	\endtrivlist\@endpefalse
}
\providecommand{\proofname}{Proof}
\fi

\usepackage{times}
\usepackage{txfonts}
\usepackage{braket}
\usepackage{colortbl}

\makeatother

\providecommand{\propositionname}{Proposition}
\providecommand{\theoremname}{Theorem}

\begin{document}
\title{Multipartite entanglement theory with entanglement-nonincreasing operations}
\author{Alexander Streltsov}
\email{streltsov.physics@gmail.com}

\affiliation{Institute of Fundamental Technological Research, ~\\
Polish Academy of Sciences, Pawi\'{n}skiego 5B, 02-106 Warsaw, Poland}
\affiliation{Centre for Quantum Optical Technologies, Centre of New Technologies,
University of Warsaw, Banacha 2c, 02-097 Warsaw, Poland}
\begin{abstract}
A key problem in quantum information science is to determine optimal
protocols for the interconversion of entangled states shared between
remote parties. While for two parties a large number of results in
this direction is available, the multipartite setting still remains
a major challenge. In this article, this problem is addressed by extending
the resource theory of entanglement for multipartite systems beyond
the standard framework of local operations and classical communication.
Specifically, we consider transformations capable of introducing a
small, controllable increase of entanglement of a state, with the
requirement that the increase can be made arbitrarily small. We demonstrate
that in this adjusted framework, the transformation rates between
multipartite states are fundamentally dictated by the bipartite entanglement
entropies of the respective quantum states. Remarkably, this approach
allows the reduction of tripartite entanglement to its bipartite analog,
indicating that every pure tripartite state can be reversibly synthesized
from a suitable number of singlets distributed between pairs of parties. 
\end{abstract}
\maketitle

\section{Introduction}

Quantum entanglement, one of the most intriguing and fundamental phenomena
in quantum mechanics, has been extensively studied for its potential
to revolutionize our understanding of the physical world and to bring
forth a new era of technological advancements. When two distant parties
share an entangled quantum state, they acquire the ability to execute
certain tasks that would be unattainable without this type of correlations~\citep{RevModPhys.81.865}.
This unique feature of quantum entanglement dramatically expands the
operational capabilities of remote parties, enabling phenomena such
as quantum teleportation~\citep{PhysRevLett.70.1895} and quantum
key distribution~\citep{PhysRevLett.67.661} that fundamentally transcend
the limitations of classical physics. 

While the capabilities and limitations associated with entanglement
in a two-party setup have been extensively studied~\citep{RevModPhys.81.865},
the emergence of large-scale quantum networks necessitates an understanding
of states entangled across multiple parties. The significance of such
multipartite entanglement is underscored by its role in various quantum
protocols such as multipartite remote state preparation~\citep{PhysRevA.105.042611}
and quantum secret sharing~\citep{PhysRevLett.83.648,PhysRevA.59.1829}.
In the latter protocol, a confidential message is disseminated among
several parties in a manner that requires their collective cooperation
for the retrieval of the message. To fully leverage the extensive
potential of multipartite entangled states, it is essential to gain
a comprehensive understanding of how these states can be manipulated. 

Local operations and classical communication (LOCC) present a fundamental
operational framework in the exploration of quantum entanglement~\citep{PhysRevA.54.3824,RevModPhys.81.865,Chitambar2014}.
At its core, LOCC involves two distinct classes of actions. The first
class, local operations, pertains to actions performed independently
on each subsystem of an entangled state. These actions encompass unitary
transformations, measurements, and the incorporation of ancillary
systems, which can be performed locally by the distant parties. The
second class, classical communication, enables the remote parties
to distribute the results of their local operations through classical
channels. In essence, LOCC can be considered the most comprehensive
class of protocols that can be executed by remote parties without
the need for exchanging quantum particles. Consequently, LOCC protocols
cannot create entanglement, thereby rendering any entangled state
a valuable resource within this framework. 

In the bipartite setting, our comprehension of the potential and limitations
inherent in LOCC is notably advanced, especially regarding transformations
involving pure states. Given any pair of pure states $\ket{\psi}$
and $\ket{\phi}$, shared between Alice and Bob, it is possible to
verify whether the transformation $\ket{\psi}\rightarrow\ket{\phi}$
is feasible under LOCC~\citep{PhysRevLett.83.436}. Furthermore,
in the asymptotic regime --- where many copies of $\ket{\psi}$ are
at our disposal --- we have precise knowledge of the transformation
rates, which are intricately tied to the entanglement entropies of
the involved quantum states~\citep{PhysRevA.53.2046}.

In contrast to the well-established findings in bipartite systems,
the multipartite setting presents a substantially more intricate landscape.
Even when considering pure states of a few qubits, the current understanding
is characterized by isolated results~\citep{PhysRevA.62.062314,PhysRevA.63.012307,PhysRevA.72.042325,PhysRevA.95.012323,PhysRevLett.122.120503,PhysRevLett.125.080502,PhysRevA.81.012317,10.1063/1.3481573,TAJIMA20131}.
Given two arbitrary four-qubit states $\ket{\psi}^{ABCD}$ and $\ket{\phi}^{ABCD}$
our current knowledge neither permits us to verify conclusively whether
$\ket{\psi}$ can be transformed into $\ket{\phi}$ via LOCC with
unit probability, nor allows us to determine the optimal asymptotic
transformation rate for such a conversion. Needless to say, this situation
presents a compelling challenge. 

Nevertheless, recent research in the field has revealed a promising
strategy: some problems in quantum information science can be effectively
addressed by carefully relaxing conventional constraints. Notable
instances include entanglement catalysis~\citep{PhysRevLett.83.3566},
where a meaningful relaxation of the standard restrictions has been
shown to considerably simplify the problem under investigation~\citep{PhysRevLett.127.150503,PhysRevLett.127.080502,PhysRevLett.129.120506,ganardi2023catalytic,lami2023catalysis,datta2022catalysis}.
Another intriguing example explores the overlap between theories of
entanglement and thermodynamics, with a particular focus on probing
the potential existence of a principle in entanglement theory that
parallels the second law of thermodynamics. This investigation essentially
hinges on the query of reversibility in transformations between entangled
states. In other words, it questions the feasibility of performing
lossless transformations between any given entangled state $\rho$
and another entangled state $\sigma$ when considered in an asymptotic
context. In the bipartite setting, the resource theory of entanglement
is inherently asymptotically irreversible under LOCC~\citep{PhysRevLett.86.5803}.
This irreversibility persists even when the LOCC framework is relaxed,
provided that the modified set of operations still remains incapable
of generating entanglement~\citep{Lami2023}. However, there are
indications that reversibility can be established when considering
broader classes of operations, e.g., those capable of generating small,
controllable amounts of entanglement~\citep{Brandao2008,Brandao2010,Brandao2010b,berta2023gap},
or involving ancillary particles acting as an entanglement battery~\citep{ganardi2024second}. 

In this article, we propose a relaxation of the LOCC framework for
the multipartite setting. Specifically, we explore transformations
that can increase the relative entropy of entanglement of any state
by a small amount $\varepsilon$, with the requirement that $\varepsilon$
can be made arbitrarily small. This relaxation of the LOCC framework
uncovers an interesting pattern; the transformation rates between
multipartite states are intrinsically governed by the bipartite entanglement
entropies of the quantum states involved. Implications of these findings
are also discussed.

\section{\protect\label{sec:Preliminaries}Preliminaries}

Here, we introduce the notation and definitions used throughout this
article. In general, a state of a quantum system is described by a
density matrix $\rho$, i.e., a positive semi-definite matrix acting
on a Hilbert space $\mathcal{H}$. An important quantity in quantum
information theory is the von Neumann entropy, which is defined as
\begin{equation}
S(\rho)=-\mathrm{Tr}(\rho\log_{2}\rho).
\end{equation}
For two quantum states $\rho$ and $\sigma$ on the same Hilbert space,
the distance between the states can be quantified as $||\rho-\sigma||_{1}$
with the trace norm 
\begin{equation}
||M||_{1}=\mathrm{Tr}\sqrt{M^{\dagger}M}.
\end{equation}
It holds that $||\rho-\sigma||_{1}\geq0$ with equality if and only
if $\rho=\sigma$. Other useful quantities in this context are the
fidelity~\citep{UHLMANN1976273,doi:10.1080/09500349414552171}
\begin{equation}
F(\rho,\sigma)=\left(\mathrm{Tr}\sqrt{\sqrt{\rho}\sigma\sqrt{\rho}}\right)^{2},
\end{equation}
 the Bures distance 
\begin{equation}
D_{B}\left(\rho,\sigma\right)=\sqrt{2-2\sqrt{F(\rho,\sigma)}},
\end{equation}
and the quantum relative entropy~\citep{Umegaki}
\begin{equation}
S(\rho||\sigma)=\mathrm{Tr}(\rho\log_{2}\rho)-\mathrm{Tr}(\rho\log_{2}\sigma).
\end{equation}
It holds that $D_{B}(\rho,\sigma)\geq0$ and $S(\rho||\sigma)\geq0$,
with equality in both cases if and only if $\rho=\sigma$. Moreover,
the Bures distance fulfills the triangle inequality, i.e., for any
three quantum states $\rho$, $\sigma$, and $\tau$ it holds that
\begin{equation}
D_{B}\left(\rho,\tau\right)\leq D_{B}\left(\rho,\sigma\right)+D_{B}\left(\sigma,\tau\right).
\end{equation}

For two parties, Alice and Bob, the joint Hilbert space is defined
as a tensor product of the individual Hilbert spaces: $\mathcal{H}^{AB}=\mathcal{H}^{A}\otimes\mathcal{H}^{B}$.
For a quantum state $\rho^{AB}$ on the total Hilbert space $\mathcal{H}^{AB}$,
the reduced state on Alice's part is defined as $\rho^{A}=\mathrm{Tr}_{B}[\rho^{AB}]$,
where $\mathrm{Tr}_{B}$ denotes the partial trace. A quantum state
$\rho^{AB}$ is called separable~\citep{PhysRevA.40.4277,RevModPhys.81.865}
if it can be written as 
\[
\rho^{AB}=\sum_{i}p_{i}\rho_{i}^{A}\otimes\rho_{i}^{B}
\]
with probabilities $p_{i}$ and quantum states $\rho_{i}^{A}$ and
$\rho_{i}^{B}$ on $\mathcal{H}^{A}$ and $\mathcal{H}^{B}$, respectively.
Any state which is not separable is called entangled~\citep{PhysRevA.40.4277,RevModPhys.81.865}.
This definition can be extended to more than $2$ parties in a straightforward
way. A tripartite state $\rho^{ABC}$ is called fully separable if
it can be written as $\rho^{ABC}=\sum_{i}p_{i}\rho_{i}^{A}\otimes\rho_{i}^{B}\otimes\rho_{i}^{C}$,
and similarly for any number of parties $n$.

There are many possible ways to quantify the amount of entanglement
in a quantum state~\citep{RevModPhys.81.865}. In this article we
will use the generalized robustness of entanglement $R_{\mathrm{g}}$
and the relative entropy of entanglement $E_{\mathrm{r}}$ defined
as~\citep{PhysRevA.59.141,PhysRevA.67.054305,PhysRevA.68.012308,PhysRevLett.78.2275}
\begin{align}
R_{\mathrm{g}}(\rho) & =\min_{\sigma}\left\{ s\geq0:\frac{\rho+s\sigma}{1+s}\in\mathcal{S}\right\} ,\\
E_{\mathrm{r}}(\rho) & =\min_{\sigma_{s}\in\mathcal{S}}S(\rho||\sigma_{s}),
\end{align}
where $\mathcal{S}$ is the set of bipartite separable states. A closely
related entanglement quantifier which will also be used in the following
is the regularized relative entropy of entanglement 
\begin{equation}
E_{\infty}(\rho)=\lim_{n\rightarrow\infty}\frac{1}{n}E_{\mathrm{r}}(\rho^{\otimes n}).
\end{equation}

A quantum operation describes the most general transformation that
a quantum system can undergo. Quantum operations correspond to completely
positive trace preserving maps $\Lambda[\rho]=\sum_{i}K_{i}\rho K_{i}^{\dagger}$with
Kraus operators $K_{i}$ having the property $\sum_{i}K_{i}^{\dagger}K_{i}=\openone$.
An important class of operations in the bipartite setting is known
as local operations and classical communication (LOCC)~\citep{PhysRevA.54.3824,Chitambar2014}.
As mentioned in the introduction, these are transformations which
can be implemented by local actions on each of the systems and a classical
communication channel. Any entanglement $E$ measure does not increase
under LOCC~\citep{plenio2006introductionentanglementmeasures,RevModPhys.81.865},
i.e., 
\begin{equation}
E(\Lambda[\rho])\leq E(\rho)
\end{equation}
for any LOCC protocol $\Lambda$. This also applies for $E_{\mathrm{g}}$
and $E_{\mathrm{r}}$ defined above. Another important property of
$E_{\mathrm{r}}$ which we will use in this article is its behavior
on states of the form $\sum_{i}p_{i}\rho_{i}^{AB}\otimes\ket{i}\!\bra{i}^{A'}$,
where the particle $A'$ is in Alice's lab. In particular, it holds~\citep{Horodecki2005}
\begin{equation}
E_{\mathrm{r}}\left(\sum_{i}p_{i}\rho_{i}^{AB}\otimes\ket{i}\!\bra{i}^{A'}\right)=\sum_{i}p_{i}E_{\mathrm{r}}\left(\rho_{i}^{AB}\right).\label{eq:flags}
\end{equation}

\section{Asymptotically entanglement-nonincreasing operations}

A key question since the early days of entanglement theory is to describe
all state transformations possible within the LOCC setting. In the
asymptotic setup, a transformation $\rho\rightarrow\sigma$ can be
achieved with rate $r$ if for any $\varepsilon>0$ there exists some
$n$ and an LOCC protocol $\Lambda_{\mathrm{LOCC}}$ such that 
\begin{equation}
\left\Vert \Lambda_{\mathrm{LOCC}}\left[\rho^{\otimes n}\right]-\sigma^{\otimes\left\lfloor rn\right\rfloor }\right\Vert _{1}<\varepsilon.
\end{equation}
The supremum of such feasible rates $r$ is known as the transformation
rate $R(\rho\rightarrow\sigma)$. 

As previously stated, for pure bipartite states $\ket{\psi}^{AB}$
and $\ket{\phi}^{AB}$, the transformation rate is closely related
to the entanglement entropies of the quantum states involved. Specifically,
the transformation rate from $\ket{\psi}^{AB}$ to $\ket{\phi}^{AB}$
is given by the ratio of their entanglement entropies: $R(\psi^{AB}\rightarrow\phi^{AB})=E(\psi^{AB})/E(\phi^{AB})$~\citep{PhysRevA.53.2046}.
Here, the entanglement entropy $E$ is defined as the von Neumann
entropy of the reduced state, i.e., 
\begin{equation}
E(\psi^{AB})=S(\psi^{A})=-\mathrm{Tr}(\psi^{A}\log_{2}\psi^{A}).
\end{equation}
When extending this scenario to multipartite pure states, it becomes
evident that the bipartite entanglement entropies serve as upper bounds
for the transformation rates via multipartite LOCC~\citep{PhysRevLett.125.080502}:
\begin{equation}
R(\psi\rightarrow\phi)\leq\min_{T}\frac{E^{T|\overline{T}}(\psi)}{E^{T|\overline{T}}(\phi)}=\min_{T}\frac{S(\psi^{T})}{S(\phi^{T})},
\end{equation}
where the minimum is taken over all subsets $T$ of all the parties,
and $\overline{T}$ is the complement of $T$. As we will see in the
following, this inequality can be turned into an equality, if the
set of LOCC operations is extended accordingly.

To extend the set of LOCC operations, we will consider transformations
that allow the injection of a small, controllable amount of entanglement
into the system, similar to the approach followed in~\citep{Brandao2008,Brandao2010,Brandao2010b}.
While these transformations may be more challenging to implement compared
to LOCC, they provide valuable insights into the limits of entanglement
manipulation. A guiding principle in our framework is to preserve
the resource-like nature of entanglement. To achieve this, our extension
ensures that the asymptotic transformation rates between any two entangled
states, $\rho$ and $\sigma$, remain bounded. In the multipartite
case, this implies that every multipartite entangled state allows
for a finite rate of singlet extraction between any pair of parties.
This suggests that our extension of the LOCC paradigm yields a theory
where singlets maintain their fundamental role as valuable resource
states. 

\begin{figure}
\includegraphics[width=0.5\columnwidth]{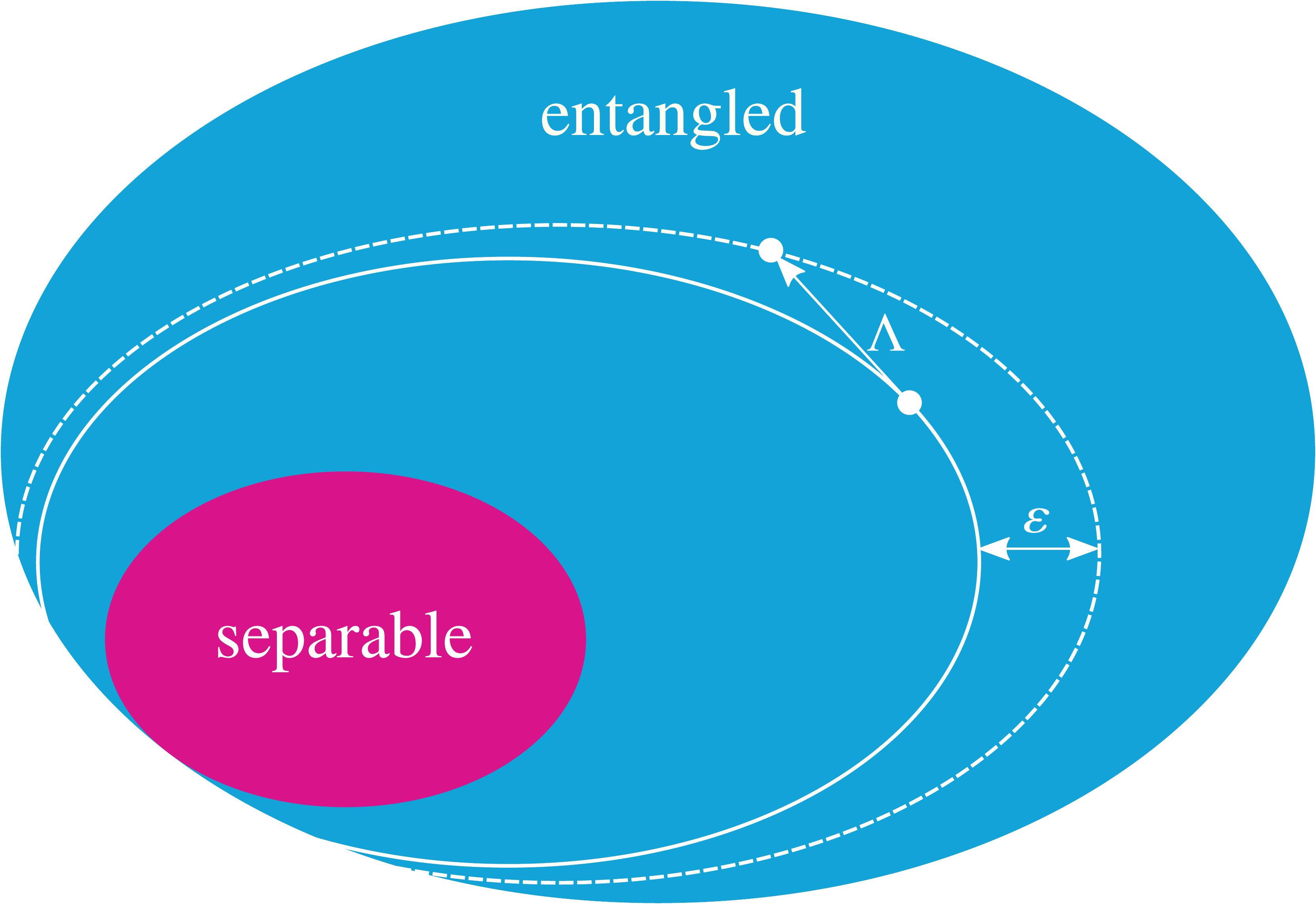}

\caption{\protect\label{fig:1}An operation $\Lambda$ is called $\varepsilon$-entanglement-nonincreasing
if it can increase the relative entropy of entanglement of any state
by at most $\varepsilon$. The figure shows states with a constant
relative entropy of entanglement $E_{\mathrm{r}}=c$ (solid curve),
and states with $E_{\mathrm{r}}=c+\varepsilon$ (dashed curve). An
$\varepsilon$-entanglement-nonincreasing operation cannot transform
states inside the solid boundary outside of the dashed boundary.}

\end{figure}
We say that a quantum operation $\Lambda$ on a bipartite Hilbert
space $\mathcal{H}^{AB}$ is \emph{$\varepsilon$-entanglement-nonincreasing}
if 
\begin{equation}
E_{\mathrm{r}}(\Lambda[\rho])-E_{\mathrm{r}}(\rho)\leq\varepsilon\label{eq:epsilon}
\end{equation}
for all states $\rho$ acting on $\mathcal{H}^{AB}$, see also Fig.~\ref{fig:1}.
A sequence of operations $\{\Lambda_{n}\}$ is called \emph{asymptotically
entanglement-nonincreasing (AEN)} if for any $\varepsilon>0$ the
operations $\Lambda_{n}$ are $\varepsilon$-entanglement-nonincreasing
for all $n$ large enough. We further say that $\rho^{AB}$ can be
converted into $\sigma^{AB}$ with rate $r$ via AEN operations if
for any $\delta>0$ there exists a monotonic infinite integer sequence
$\{k_{n}\}$ and an AEN sequence $\{\Lambda_{k}\}$ such that
\begin{equation}
\frac{1}{2}\left\Vert \Lambda_{k_{n}}\left[\rho^{\otimes k_{n}}\right]-\sigma^{\otimes\left\lfloor rk_{n}\right\rfloor }\right\Vert _{1}<\delta\label{eq:AENrate}
\end{equation}
holds true for all $n$ large enough. Here, $\Lambda_{k_{n}}$ is
a quantum operation with the input Hilbert space being $k_{n}$ copies
of $\mathcal{H}^{AB}$, and the output Hilbert space is $\left\lfloor rk_{n}\right\rfloor $
copies of $\mathcal{H}^{AB}$. The supremum of all rates achievable
in this setup will be denoted by $R_{\mathrm{AEN}}(\rho\rightarrow\sigma)$.
It is clear that 
\begin{equation}
R_{\mathrm{AEN}}(\rho\rightarrow\sigma)\geq R(\rho\rightarrow\sigma).
\end{equation}
As we show in the following proposition, $R_{\mathrm{AEN}}$ can be
upper bounded in terms of the regularized relative entropy of entanglement. 
\begin{prop}
For any two bipartite states $\rho$ and $\sigma$ the conversion
rate via asymptotically entanglement-nonincreasing operations is bounded
as
\begin{equation}
R_{\mathrm{AEN}}(\rho\rightarrow\sigma)\leq\frac{E_{\infty}\left(\rho\right)}{E_{\infty}\left(\sigma\right)}.\label{eq:UpperBound}
\end{equation}
\end{prop}
\begin{proof}
Let $r$ be a feasible rate for the conversion $\rho\rightarrow\sigma$,
i.e., for any $\delta>0$ there exists an integer sequence $\{k_{n}\}$
and an AEN sequence $\{\Lambda_{k}\}$ such that Eq.~(\ref{eq:AENrate})
holds for all $n$ large enough. Using asymptotic continuity of the
relative entropy of entanglement~\citep{DONALD1999257,Winter2016}
we obtain 
\begin{equation}
E_{\mathrm{r}}\left(\sigma^{\otimes\left\lfloor rk_{n}\right\rfloor }\right)-E_{\mathrm{r}}\left(\Lambda_{k_{n}}\left[\rho^{\otimes k_{n}}\right]\right)\leq\delta\log_{2}d^{rk_{n}}+(1+\delta)h\left(\frac{\delta}{1+\delta}\right),
\end{equation}
where $d=d_{A}=d_{B}$ denotes the local dimension of $\rho=\rho^{AB}$
and $\sigma=\sigma^{AB}$. Using the fact that the sequence $\{\Lambda_{k}\}$
is AEN we further have 
\begin{equation}
E_{\mathrm{r}}\left(\sigma^{\otimes\left\lfloor rk_{n}\right\rfloor }\right)\leq E_{\mathrm{r}}\left(\rho^{\otimes k_{n}}\right)+\varepsilon+\delta\log_{2}d^{rk_{n}}+(1+\delta)h\left(\frac{\delta}{1+\delta}\right)
\end{equation}
for all $n$ large enough. Dividing both sides of the inequality by
$k_{n}$ and taking the limit $n\rightarrow\infty$ we find 
\begin{equation}
rE_{\infty}\left(\sigma\right)\leq E_{\infty}\left(\rho\right)+\delta r\log_{2}d.
\end{equation}
Since we can choose arbitrary $\delta>0$, it follows that $r\leq E_{\infty}\left(\rho\right)/E_{\infty}\left(\sigma\right)$
and the proof is complete.
\end{proof}
An important implication of this result is that the AEN setting ensures
a bounded singlet distillation rate, as $R_{\mathrm{AEN}}(\rho\rightarrow\psi^{-})\leq E_{\infty}(\rho)$
for the singlet state $\ket{\psi^{-}}=(\ket{01}-\ket{10})/\sqrt{2}$. 

It is useful to compare the class of operations discussed above with
those introduced in~\citep{Brandao2008}, where the authors examined
$\varepsilon$-non-entangling operations in the bipartite setting.
These are operations $\Lambda$ that can transform a separable state
$\rho_{s}$ into an entangled state such that $R_{\mathrm{g}}(\Lambda[\rho_{s}])\leq\varepsilon$,
where $R_{\mathrm{g}}$ denotes the generalized robustness of entanglement
which has been defined in Section~\ref{sec:Preliminaries}. The authors
of~\citep{Brandao2008} also defined a sequence of operations $\{\Lambda_{n}\}$
to be asymptotically non-entangling if each $\Lambda_{n}$ is $\varepsilon_{n}$-non-entangling
and $\lim_{n\rightarrow\infty}\varepsilon_{n}=0$. We observe that
any $\varepsilon$-non-entangling operation is also $\varepsilon$-entanglement-nonincreasing.
Furthermore, a sequence of operations that is asymptotically non-entangling
is also AEN. However, it remains an open question whether these two
notions coincide, i.e., whether every AEN sequence is also asymptotically
non-entangling. 

\section{AEN in multipartite systems}

For multipartite systems, we will consider operations which are $\varepsilon$-entanglement-nonincreasing
in any bipartition. As an example, for three parties, $A$, $B$,
and $C$, we require that Eq.~(\ref{eq:epsilon}) is fulfilled for
the bipartite relative entropy of entanglement in all cuts $A|BC$,
$B|AC$, and $C|AB$. This definition is then extended to more than
3 parties in a straightforward way, i.e., it is required that
\begin{equation}
E_{\mathrm{r}}^{T|\overline{T}}(\Lambda[\rho])-E_{\mathrm{r}}^{T|\overline{T}}(\rho)\leq\varepsilon
\end{equation}
holds true for any bipartition $T|\overline{T}$ of the total multipartite
system. Correspondingly, a sequence of multipartite operations $\{\Lambda_{n}\}$
will be called asymptotically entanglement-nonincreasing, if it is
AEN in any bipartition. 

Asymptotic conversion rates via multipartite AEN are then defined
analogously to the bipartite setting. From Eq.~(\ref{eq:UpperBound})
it is clear that the optimal rate is not larger than 
\begin{equation}
R_{\mathrm{AEN}}(\rho\rightarrow\sigma)\leq\min_{T}\frac{E_{\infty}^{T|\overline{T}}(\rho)}{E_{\infty}^{T|\overline{T}}(\sigma)}.\label{eq:UpperBoundMultipartite}
\end{equation}
This means that multipartite AEN operations lead to bounded distillation
rates, when it comes to distilling singlets between any pair of parties.
Since there exist multipartite entangled states which are separable
in all bipartitions~\citep{DiVincenzo2003}, AEN operations can in
principle create large amounts of multipartite entanglement. Nevertheless,
Eq.~(\ref{eq:UpperBoundMultipartite}) ensures that the theory obtained
in this way does not trivialize, preserving the valuable role of entanglement
as a resource.

Equipped with these tools, we are ready to present the main result
of this article. 
\begin{thm}
\label{thm:Main}For any two multipartite pure states $\ket{\psi}$
and $\ket{\phi}$ the optimal conversion rate via asymptotically entanglement-nonincreasing
operations is given by
\begin{equation}
R_{\mathrm{AEN}}(\psi\rightarrow\phi)=\min_{T}\frac{E^{T|\overline{T}}(\psi)}{E^{T|\overline{T}}(\phi)}=\min_{T}\frac{S(\psi^{T})}{S(\phi^{T})}.
\end{equation}
\end{thm}
\begin{proof}
From Eq.~(\ref{eq:UpperBoundMultipartite}) we see that the maximal
rate cannot exceed $\min_{T}S(\psi^{T})/S(\phi^{T})$. We will now
present a sequence of AEN operations which achieves conversion at
this rate. Consider the sequence
\begin{align}
\Lambda_{n}^{\psi,\phi}\left[\rho\right] & =\mathrm{Tr}\left[\ket{\psi}\!\bra{\psi}^{\otimes n}\rho\right]\ket{\phi}\!\bra{\phi}^{\otimes\left\lfloor rn\right\rfloor }\otimes\ket{0}\!\bra{0}^{K}+\mathrm{Tr}\left[\left(\openone-\ket{\psi}\!\bra{\psi}^{\otimes n}\right)\rho\right]\mu_{s}\otimes\ket{1}\!\bra{1}^{K},
\end{align}
with some fully separable state $\mu_{s}$, and the register $K$
is in possession of Alice. When applied onto the state $\rho=\psi^{\otimes n}$,
we obtain
\begin{equation}
\Lambda_{n}^{\psi,\phi}\left[\psi^{\otimes n}\right]=\phi^{\otimes\left\lfloor rn\right\rfloor }\otimes\ket{0}\!\bra{0}^{K}.
\end{equation}
Tracing out the register $K$, we see that $\ket{\psi}$ can be converted
into $\ket{\phi}$ with rate $r$ in this setting. We show in Proposition~\ref{prop:Main}
that in the bipartite setting the sequence $\Lambda_{n}^{\psi,\phi}$
is AEN whenever $r<S(\psi^{A})/S(\phi^{A})$. This means that in the
multipartite setting, the sequence $\Lambda_{n}^{\psi,\phi}$ is AEN
in the bipartition $T|\overline{T}$ whenever $r<S(\psi^{T})/S(\phi^{T})$.
Thus, for 
\begin{equation}
r<\min_{T}\frac{S(\psi^{T})}{S(\phi^{T})}\label{eq:q}
\end{equation}
 the sequence $\Lambda_{n}^{\psi,\phi}$ is AEN in any bipartition.
In summary, this shows that it is possible to convert $\ket{\psi}$
into $\ket{\phi}$ via AEN operations at any rate $r$ which fulfills
Eq.~(\ref{eq:q}). 
\end{proof}
The above theorem leads to several implications, which we will discuss
in the following. Within the LOCC framework, an asymptotic transformation
between two states, $\rho$ and $\sigma$, is dubbed reversible if
the transformation rates satisfy the condition
\begin{equation}
R(\rho\rightarrow\sigma)=R(\sigma\rightarrow\rho)^{-1}.\label{eq:Reversibility}
\end{equation}
As mentioned earlier, a salient characteristic of entanglement theory
is the general irreversibility of asymptotic transformations via LOCC~\citep{PhysRevLett.86.5803},
implying the existence of states that violate Eq.~(\ref{eq:Reversibility}).
Though in the bipartite context, asymptotic LOCC transformations between
pure entangled states are reversible~\citep{PhysRevA.53.2046}, this
ceases to hold when extended to scenarios involving more than two
parties, even for pure states~\citep{PhysRevA.63.012307}.

In the AEN setting, we say that a reversible transformation between
$\rho$ and $\sigma$ is possible if $R_{\mathrm{AEN}}(\rho\rightarrow\sigma)=R_{\mathrm{AEN}}(\sigma\rightarrow\rho)^{-1}$.
From Theorem~\ref{thm:Main} we immediately see that a reversible
transformation between multipartite pure states is possible if and
only if 
\begin{equation}
\frac{S(\psi^{T})}{S(\phi^{T})}=\frac{S(\psi^{T'})}{S(\phi^{T'})}\label{eq:ReversibilityMultipartite}
\end{equation}
for any two subsets of all the parties $T$ and $T'$. This directly
implies that asymptotic transformations between multipartite pure
states are irreversible in general even under AEN operations. 

On the other hand, reversibility can be established for some important
classes of states. In particular, our framework enables reversible
transformations between GHZ and $W$ states, i.e., $3$-qubit states
of the form
\begin{align}
\ket{W} & =\frac{1}{\sqrt{3}}(\ket{001}+\ket{010}+\ket{100}),\\
\ket{\mathrm{GHZ}} & =\frac{1}{\sqrt{2}}(\ket{000}+\ket{111}).
\end{align}
By symmetry of the states, it is straightforward to see that Eq.~(\ref{eq:ReversibilityMultipartite})
is fulfilled in this case. For the transformation rate we obtain 
\begin{equation}
R_{\mathrm{AEN}}(\ket{W}\rightarrow\ket{\mathrm{GHZ}})=h(1/3)\approx0.92
\end{equation}
with the binary entropy $h(x)=-x\log_{2}x-(1-x)\log_{2}(1-x)$. Note
that a reversible conversion between these states is not possible
via LOCC, and even in an extended setting allowing for all operations
which preserve the positivity of the partial transpose in any bipartition~\citep{PhysRevA.72.042325}. 

Another interesting case when the AEN setting allows for reversible
transformations is a conversion between two GHZ states and three singlets
in a tripartite configuration. In more detail, each of the three parties,
Alice, Bob and Charlie, are holding two qubits and wish to convert
the following states into each other:
\begin{align}
\ket{\psi}^{ABC} & =\ket{\mathrm{GHZ}}^{A_{1}B_{1}C_{1}}\otimes\ket{\mathrm{GHZ}}^{A_{2}B_{2}C_{2}},\\
\ket{\phi}^{ABC} & =\ket{\psi^{-}}^{A_{1}B_{2}}\otimes\ket{\psi^{-}}^{B_{1}C_{2}}\otimes\ket{\psi^{-}}^{C_{1}A_{2}}.
\end{align}
Also in this case it is straightforward to see that Eq.~(\ref{eq:ReversibilityMultipartite})
holds, and moreover in the AEN setting these states can be asymptotically
interconverted with unit rate, as follows directly from Theorem~\ref{thm:Main}.
Note that a reversible interconversion between these states is not
possible in the LOCC framework~\citep{PhysRevA.63.012307,linden1999reversibility}.

\begin{figure*}
\includegraphics[width=0.75\paperwidth]{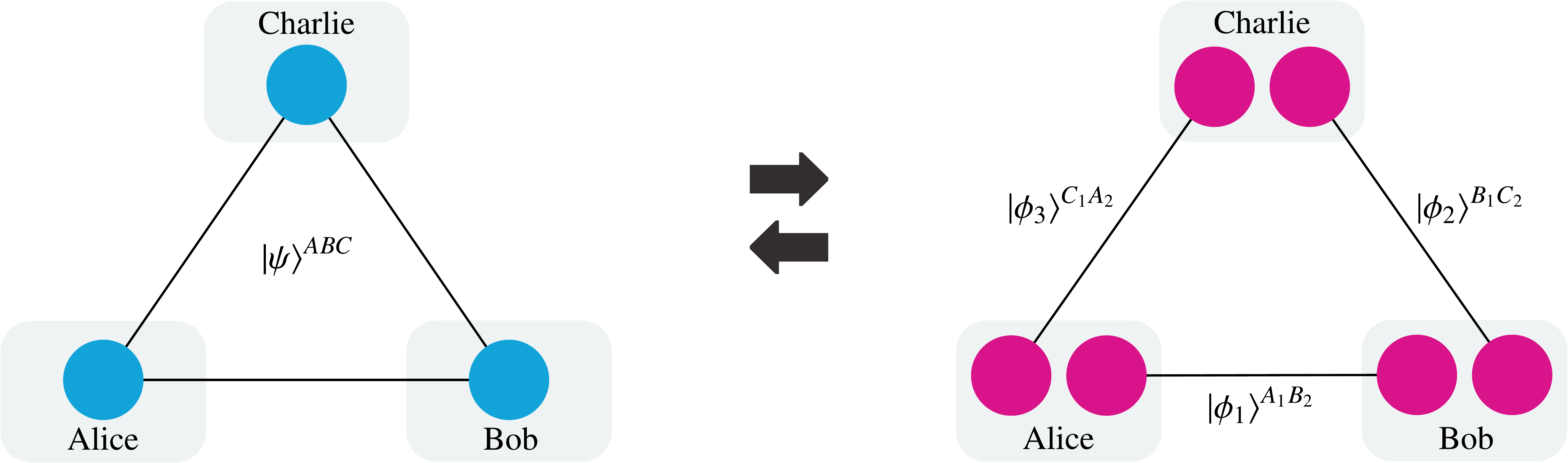}

\caption{\protect\label{fig:REG}Through the application of AEN operations,
we can execute a reversible conversion of any tripartite pure state
$\ket{\psi}^{ABC}$ (as shown in the left part of the figure) into
a pure state consisting of bipartite entangled states between each
pair of parties (represented in the right part of the figure).}
\end{figure*}
We further notice that in the tripartite AEN setting singlets shared
by each of the two parties form a reversible entanglement generating
set (REGS). In more detail, we will show that for any tripartite state
$\ket{\psi}^{ABC}$ and any $\varepsilon>0$ there exists a state
$\ket{\psi_{\varepsilon}}^{ABC}$ such that $|\!\braket{\psi|\psi_{\varepsilon}}\!|^{2}>1-\varepsilon$
and integers $n$, $m_{1}$, $m_{2}$ and $m_{3}$ such that 
\begin{equation}
R_{\mathrm{AEN}}\left(\ket{\psi_{\varepsilon}}^{\otimes n}\rightarrow\ket{\psi^{-}}_{A_{1}B_{2}}^{\otimes m_{1}}\otimes\ket{\psi^{-}}_{B_{1}C_{2}}^{\otimes m_{2}}\otimes\ket{\psi^{-}}_{C_{1}A_{2}}^{\otimes m_{3}}\right)=R_{\mathrm{AEN}}\left(\ket{\psi^{-}}_{A_{1}B_{2}}^{\otimes m_{1}}\otimes\ket{\psi^{-}}_{B_{1}C_{2}}^{\otimes m_{2}}\otimes\ket{\psi^{-}}_{C_{1}A_{2}}^{\otimes m_{3}}\rightarrow\ket{\psi_{\varepsilon}}^{\otimes n}\right)=1.
\end{equation}
This means that every tripartite pure state can be approximated by
a state which is reversibly interconvertible into singlets shared
between the parties. We note that the concept of REGS has been first
discussed in~\citep{PhysRevA.63.012307}.

To prove this, we will first show that any tripartite state $\ket{\psi}^{ABC}$
is reversibly interconvertible into a state comprising bipartite entangled
states shared between each pair of parties (see also Fig.~\ref{fig:REG}).
In more detail, we will show that for any tripartite pure state $\ket{\psi}^{ABC}$
there exists a state of the form 
\begin{equation}
\ket{\phi}^{ABC}=\ket{\phi_{1}}^{A_{1}B_{2}}\otimes\ket{\phi_{2}}^{B_{1}C_{2}}\otimes\ket{\phi_{3}}^{C_{1}A_{2}}
\end{equation}
having the same von Neumann entropies of the subsystems as $\ket{\psi}^{ABC}$,
i.e., \begin{subequations}
\begin{align}
S(\psi^{A}) & =S(\phi^{A}),\\
S(\psi^{B}) & =S(\phi^{B}),\\
S(\psi^{C}) & =S(\phi^{C}).
\end{align}
\end{subequations} For this, let us denote the entanglement entropies
of the states $\ket{\phi_{i}}$ by $s_{i}=E(\phi_{i})$. For a given
state $\ket{\psi}^{ABC}$ there exists a state $\ket{\phi}^{ABC}$
with the claimed features whenever there exist nonnegative numbers
$s_{i}$ fulfilling the conditions \begin{subequations}
\begin{align}
S(\psi^{A}) & =s_{1}+s_{3},\\
S(\psi^{B}) & =s_{1}+s_{2},\\
S(\psi^{C}) & =s_{2}+s_{3}.
\end{align}
\end{subequations} Solving these equations for $s_{i}$ we obtain
\begin{subequations}\label{eq:REG}
\begin{align}
s_{1} & =\frac{1}{2}\left[S(\psi^{A})+S(\psi^{B})-S(\psi^{C})\right],\\
s_{2} & =\frac{1}{2}\left[S(\psi^{B})+S(\psi^{C})-S(\psi^{A})\right],\\
s_{3} & =\frac{1}{2}\left[S(\psi^{A})+S(\psi^{C})-S(\psi^{B})\right].
\end{align}
\end{subequations} By the subadditivity of the von Neumann entropy
we immediately see that each term on the right-hand side of Eqs.~(\ref{eq:REG})
is nonnegative. This proves that the state $\ket{\phi}^{ABC}$ with
the claimed features exists for any tripartite pure state $\ket{\psi}^{ABC}$.
Using Theorem~\ref{thm:Main}, we see that the conversion rates fulfill
\begin{equation}
R_{\mathrm{AEN}}(\ket{\psi}\rightarrow\ket{\phi})=R_{\mathrm{AEN}}(\ket{\phi}\rightarrow\ket{\psi})=1
\end{equation}
which means that $\ket{\psi}^{ABC}$ and $\ket{\phi}^{ABC}$ are reversibly
interconvertible via AEN operations. 

Assume now that the state $\ket{\psi}^{ABC}$ is chosen such that
$s_{1}$, $s_{2}$, and $s_{3}$ in Eqs.~(\ref{eq:REG}) are rational.
Then, there exists an integer $n$ such that $ns_{1}$, $ns_{2}$,
and $ns_{3}$ are integers. Moreover, Theorem~\ref{thm:Main} directly
implies that $\ket{\psi}^{\otimes n}$ can be reversibly interconverted
into the state 
\begin{equation}
\ket{\phi'}=\ket{\psi^{-}}_{A_{1}B_{2}}^{\otimes ns_{1}}\otimes\ket{\psi^{-}}_{B_{1}C_{2}}^{\otimes ns_{2}}\otimes\ket{\psi^{-}}_{C_{1}A_{2}}^{\otimes ns_{3}}
\end{equation}
 via AEN operations, i.e., 
\begin{equation}
R_{\mathrm{AEN}}(\ket{\psi}^{\otimes n}\rightarrow\ket{\phi'})=R_{\mathrm{AEN}}(\ket{\phi'}\rightarrow\ket{\psi}^{\otimes n})=1.
\end{equation}
This proves that singlets form a REGS in this case.

If $s_{1}$, $s_{2}$, and $s_{3}$ are not rational, then for any
$\varepsilon>0$ there exists a state $\ket{\psi_{\varepsilon}}^{ABC}$
such that $|\!\braket{\psi|\psi_{\varepsilon}}\!|^{2}>1-\varepsilon$
and \begin{subequations}
\begin{align}
s_{1}' & =\frac{1}{2}\left[S(\psi_{\varepsilon}^{A})+S(\psi_{\varepsilon}^{B})-S(\psi_{\varepsilon}^{C})\right],\\
s_{2}' & =\frac{1}{2}\left[S(\psi_{\varepsilon}^{B})+S(\psi_{\varepsilon}^{C})-S(\psi_{\varepsilon}^{A})\right],\\
s_{3}' & =\frac{1}{2}\left[S(\psi_{\varepsilon}^{A})+S(\psi_{\varepsilon}^{C})-S(\psi_{\varepsilon}^{B})\right]
\end{align}
 \end{subequations} are all rational. Similarly as above, it follows
that 
\begin{equation}
R_{\mathrm{AEN}}(\ket{\psi_{\varepsilon}}^{\otimes n}\rightarrow\ket{\phi'})=R_{\mathrm{AEN}}(\ket{\phi'}\rightarrow\ket{\psi_{\varepsilon}}^{\otimes n})=1
\end{equation}
holds true for some $n$. This proves that singlets form a REGS for
AEN transformations in the tripartite setting. Note that in the LOCC
framework singlets are not enough to generate all states reversibly
for more than two parties~\citep{PhysRevA.63.012307,linden1999reversibility}. 

At this stage, one might wonder whether the choice of the relative
entropy of entanglement in Eq. (\ref{eq:epsilon}) is essential, or
if other entanglement measures could similarly yield a result analogous
to Theorem~\ref{thm:Main}. Although we cannot fully resolve this
question at present, we note that the proof of our main result in
Theorem~\ref{thm:Main} relies on several specific properties of
the relative entropy of entanglement. Since other entanglement quantifiers
may not share all of the properties used in the proofs, it appears
unlikely that our approach can be directly extended to other entanglement
measures.

\section{\protect\label{sec:Proofs}Methods}

In this section we give more details on the technical methods used
in the proof of Theorem~\ref{thm:Main}. The key element of the proof
is Proposition~\ref{prop:Main}. As a preparation for the proof of
Proposition~\ref{prop:Main} we need the following results.

We will first consider the following sequence of transformations:
\begin{align}
\Lambda_{n}\left[\rho\right] & =\mathrm{Tr}\left[\ket{\psi^{-}}\!\bra{\psi^{-}}^{\otimes n}\rho\right]\ket{\psi^{-}}\!\bra{\psi^{-}}^{\otimes\left\lfloor rn\right\rfloor }\otimes\ket{0}\!\bra{0}^{K}+\mathrm{Tr}\left[\left(\openone-\ket{\psi^{-}}\!\bra{\psi^{-}}^{\otimes n}\right)\rho\right]\mu_{s}\otimes\ket{1}\!\bra{1}^{K}\label{eq:LambdaN}
\end{align}
with some fully separable state $\mu_{s}$, and $K$ is a register
in the possession of Alice. We will first show that this transformation
generates a vanishingly small amount of entanglement for large $n$.
As entanglement quantifiers we will use the generalized robustness
of entanglement $R_{\mathrm{g}}$ and the relative entropy of entanglement
$E_{\mathrm{r}}$ defined in Section~\ref{sec:Preliminaries}.
\begin{prop}
\label{prop:RobustnessSeparable}For any separable state $\rho_{s}$
and any $r>0$ it holds that 
\begin{equation}
R_{\mathrm{g}}\left(\Lambda_{n}\left[\rho_{s}\right]\right)\leq\frac{1}{2^{n(1-r)}}.
\end{equation}
\end{prop}
\begin{proof}
By convexity of the generalized robustness of entanglement~\citep{PhysRevA.67.054305}
we obtain the following bound for any separable state $\rho_{s}$:
\begin{equation}
R_{\mathrm{g}}\left(\Lambda_{n}\left[\rho_{s}\right]\right)\leq\mathrm{Tr}\left[\ket{\psi^{-}}\!\bra{\psi^{-}}^{\otimes n}\rho_{s}\right]R_{\mathrm{g}}\left(\ket{\psi^{-}}\!\bra{\psi^{-}}^{\otimes\left\lfloor rn\right\rfloor }\right).
\end{equation}
Recalling that the generalized robustness of a pure state is given
by~\citep{PhysRevA.67.054305,PhysRevA.68.012308} 
\begin{equation}
R_{\mathrm{g}}(\psi)=\left(\sum_{i}c_{i}\right)^{2}-1,
\end{equation}
where $c_{i}$ are the Schmidt coefficients (i.e. $\ket{\psi}=\sum_{i}c_{i}\ket{ii}$),
we obtain
\begin{equation}
R_{\mathrm{g}}\left(\ket{\psi^{-}}\!\bra{\psi^{-}}^{\otimes m}\right)=\left(\sum_{i=0}^{2^{m}-1}\frac{1}{\sqrt{2^{m}}}\right)^{2}-1=2^{m}-1.
\end{equation}
Moreover, for any separable state $\rho_{s}$ it holds that~\citep{PhysRevA.60.1888}
\begin{equation}
\mathrm{Tr}\left[\ket{\psi^{-}}\!\bra{\psi^{-}}^{\otimes n}\rho_{s}\right]\leq\frac{1}{2^{n}}.
\end{equation}
Collecting these results, we arrive at the desired inequality:
\begin{align}
R_{\mathrm{g}}\left(\Lambda_{n}\left[\rho_{s}\right]\right) & \leq\mathrm{Tr}\left[\ket{\psi^{-}}\!\bra{\psi^{-}}^{\otimes n}\rho_{s}\right]R_{\mathrm{g}}\left(\ket{\psi^{-}}\!\bra{\psi^{-}}^{\otimes\left\lfloor rn\right\rfloor }\right)\leq\frac{2^{\left\lfloor rn\right\rfloor }-1}{2^{n}}\leq\frac{1}{2^{n(1-r)}}.
\end{align}
\end{proof}
We will further make use of the following proposition which has been
proven in~\citep{Brandao2010b}.
\begin{prop}
\label{prop:RelativeEntropy}For any quantum operation $\Lambda$
fulfilling 
\begin{equation}
\max_{\rho_{s}\in\mathcal{S}}R_{\mathrm{g}}\left(\Lambda\left[\rho_{s}\right]\right)\leq\varepsilon\label{eq:RobustnessDelta}
\end{equation}
the following holds for any quantum state $\rho$:
\begin{equation}
E_{\mathrm{r}}\left(\Lambda\left[\rho\right]\right)-E_{\mathrm{r}}\left(\rho\right)\leq\log_{2}\left(1+\varepsilon\right).\label{eq:RelEntDelta}
\end{equation}
\end{prop}
In the next step, we will provide an upper bound for the fidelity
between $n$ singlets and any state $\rho$ on the same Hilbert space. 
\begin{prop}
\label{prop:SingletFidelity}The following inequality holds for all
states $\rho$, all $n$, and all $r>0$:
\begin{equation}
F\left(\ket{\psi^{-}}\!\bra{\psi^{-}}^{\otimes n},\rho\right)\leq\frac{E_{\mathrm{r}}(\rho)+\log_{2}\left(1+\frac{1}{2^{n(1-r)}}\right)}{\left\lfloor rn\right\rfloor }.
\end{equation}
\end{prop}
\begin{proof}
Due to Propositions~\ref{prop:RobustnessSeparable} and \ref{prop:RelativeEntropy},
we find that 
\begin{equation}
E_{\mathrm{r}}(\Lambda_{n}[\rho])\leq E_{\mathrm{r}}(\rho)+\log_{2}\left(1+\frac{1}{2^{n(1-r)}}\right)
\end{equation}
holds for all states $\rho$, where $\Lambda_{n}$ is the transformation
defined in Eq.~(\ref{eq:LambdaN}). In the next step, observe that\footnote{Note that this equality would not hold if one would define $\Lambda_{n}$
without the register $K$.} 
\begin{equation}
E_{\mathrm{r}}(\Lambda_{n}[\rho])=\mathrm{Tr}\left[\ket{\psi^{-}}\!\bra{\psi^{-}}^{\otimes n}\rho\right]E_{\mathrm{r}}\left(\ket{\psi^{-}}\!\bra{\psi^{-}}^{\otimes\left\lfloor rn\right\rfloor }\right),
\end{equation}
which follows directly from the definition of $\Lambda_{n}$ and Eq.~(\ref{eq:flags}).
This means that 
\begin{equation}
F\left(\ket{\psi^{-}}\!\bra{\psi^{-}}^{\otimes n},\rho\right)=\mathrm{Tr}\left[\ket{\psi^{-}}\!\bra{\psi^{-}}^{\otimes n}\rho\right]\leq\frac{E_{\mathrm{r}}(\rho)+\log_{2}\left(1+\frac{1}{2^{n(1-r)}}\right)}{E_{\mathrm{r}}\left(\ket{\psi^{-}}\!\bra{\psi^{-}}^{\otimes\left\lfloor rn\right\rfloor }\right)}.
\end{equation}
The proof is complete by noting that $E_{\mathrm{r}}(\ket{\psi^{-}}\!\bra{\psi^{-}}^{\otimes\left\lfloor rn\right\rfloor })=\left\lfloor rn\right\rfloor $.
\end{proof}
Another useful tool which will be applied in this article is a variation
of Corollary 7 in~\citep{Hayashi_2003}. In particular, we consider
LOCC protocols which convert $\ket{\psi}^{\otimes n}$ into a maximally
entangled state of local dimension $L_{n}$, i.e., the target state
is 
\begin{equation}
\ket{\phi_{L_{n}}}=\frac{1}{\sqrt{L_{n}}}\sum_{i=0}^{L_{n}-1}\ket{ii}.
\end{equation}
We assume that the transformation is exact but probabilistic, having
success probability $P_{n}$. Moreover, let $p$ be a probability
distribution containing the Schmidt coefficients of $\ket{\psi}$,
i.e., $\ket{\psi}=\sum_{i}\sqrt{p_{i}}\ket{ii}$. For $r>0$ we define
the function~\citep{Hayashi_2003} 
\begin{equation}
E(r)=\min_{q:D(q||p)\leq r}\left\{ D(q||p)+H(q)\right\} ,
\end{equation}
where $q$ is a probability distribution, $H(q)=-\sum_{i}q_{i}\log_{2}q_{i}$
is the Shannon entropy, and $D(q||p)=\sum_{i}q_{i}\log_{2}\frac{q_{i}}{p_{i}}$
is the relative entropy. The following proposition is a direct consequence
of Corollary 7 in~\citep{Hayashi_2003}.
\begin{prop}
For any $r>0$, $\varepsilon>0$, and $\delta>0$ there exists a sequence
of LOCC protocols converting $\ket{\psi}^{\otimes n}$ into the state
$\ket{\phi_{L_{n}}}$ with success probability $P_{n}$ such that
the following inequalities hold for all $n$ large enough:
\begin{align}
r-\delta & \leq-\frac{1}{n}\log_{2}(1-P_{n}),\\
\frac{1}{n}\log_{2}L_{n} & \geq E(r)-\varepsilon.
\end{align}
\end{prop}
With this result, we are now able to prove the following proposition.
\begin{prop}
\label{prop:PureExponential}For any $R<S(\psi^{A})$ there exists
$\alpha>0$ and a sequence of LOCC protocols $\Phi_{n}$ such that
\begin{equation}
F\left(\Phi_{n}\left[\psi^{\otimes n}\right],\ket{\phi^{+}}\!\bra{\phi^{+}}^{\otimes\left\lfloor Rn\right\rfloor }\right)\geq1-2^{-\alpha n}\label{eq:PureExponential}
\end{equation}
for all $n$ large enough.
\end{prop}
\begin{proof}
Note that $E(r)$ is continuous and monotonically decreasing for $r>0$,
and moreover $\lim_{r\rightarrow0}E(r)=H(p)=S(\psi^{A})$~\citep{Hayashi_2003}.
By continuity, for any $\varepsilon'>0$ there exists $r_{\varepsilon'}>0$
such that $E(r_{\varepsilon'})=S(\psi^{A})-\varepsilon'$. It follows
that for any $\varepsilon>0$, $\delta>0$, and $\varepsilon'>0$
the following inequalities hold for all $n$ large enough:
\begin{align}
P_{n} & \geq1-2^{-n(r_{\varepsilon'}-\delta)},\\
\log_{2}L_{n} & \geq n\left[S(\psi^{A})-\varepsilon'-\varepsilon\right].
\end{align}
Choosing $R=S(\psi^{A})-\varepsilon'-\varepsilon$, we conclude that
for any $\varepsilon>0$, $\delta>0$, and $\varepsilon'>0$ there
exists a sequence of (deterministic) LOCC protocol $\Phi_{n}$ such
that the following inequality holds for all $n$ large enough:
\begin{equation}
F\left(\Phi_{n}\left[\psi^{\otimes n}\right],\ket{\phi^{+}}\!\bra{\phi^{+}}^{\otimes\left\lfloor Rn\right\rfloor }\right)\geq P_{n}\geq1-2^{-n(r_{\varepsilon'}-\delta)}.
\end{equation}
Since $\varepsilon>0$, $\delta>0$, and $\varepsilon'>0$ can be
chosen arbitrarily, we obtain the desired inequality~(\ref{eq:PureExponential}).
\end{proof}
We will now consider a generalization of the transformation $\Lambda_{n}$
defined above. For two pure states $\ket{\psi}$ and $\ket{\phi}$
and some $q>0$ we define
\begin{align}
\Lambda_{n}^{\psi,\phi}\left[\rho\right] & =\mathrm{Tr}\left[\ket{\psi}\!\bra{\psi}^{\otimes n}\rho\right]\ket{\phi}\!\bra{\phi}^{\otimes\left\lfloor qn\right\rfloor }\otimes\ket{0}\!\bra{0}^{K}+\mathrm{Tr}\left[\left(\openone-\ket{\psi}\!\bra{\psi}^{\otimes n}\right)\rho\right]\mu_{s}\otimes\ket{1}\!\bra{1}^{K},
\end{align}
where $\mu_{s}$ is some fully separable state and the register $K$
is in possession of Alice. Note that this is the type of transformations
which is used in the proof of Theorem~\ref{thm:Main}. The following
proposition is the key element of the proof of Theorem~\ref{thm:Main}.
\begin{prop}
\label{prop:Main}For any two bipartite pure states $\ket{\psi}$
and $\ket{\phi}$ and any $q<S(\psi^{A})/S(\phi^{A})$ the sequence
of operations $\Lambda_{n}^{\psi,\phi}$ is asymptotically entanglement-nonincreasing.
\end{prop}
\begin{proof}
The relative entropy of entanglement of $\Lambda_{n}^{\psi,\phi}[\rho]$
is given by 
\begin{equation}
E_{\mathrm{r}}\left(\Lambda_{n}^{\psi,\phi}[\rho]\right)=\mathrm{Tr}\left[\ket{\psi}\!\bra{\psi}^{\otimes n}\rho\right]\times\left\lfloor qn\right\rfloor S(\phi^{A}).\label{eq:Er}
\end{equation}
Consider now a sequence of states $\rho_{n}$ acting on $n$ copies
of $\mathcal{H}^{AB}$. Our goal is to show that for any such sequence
and any $\varepsilon>0$ it holds that
\begin{equation}
E_{\mathrm{r}}\left(\Lambda_{n}^{\psi,\phi}[\rho_{n}]\right)-E_{\mathrm{r}}\left(\rho_{n}\right)<\varepsilon\label{eq:Sequence}
\end{equation}
for all $n$ large enough. 

Choose some $R>0$. By Proposition~\ref{prop:SingletFidelity}, for
any sequence of states $\sigma_{n}$, any $r>0$ and any $n$ it holds
that
\begin{equation}
F\left(\ket{\psi^{-}}\!\bra{\psi^{-}}^{\otimes\left\lfloor Rn\right\rfloor },\sigma_{n}\right)\leq\frac{E_{\mathrm{r}}(\sigma_{n})+\log_{2}\left(1+\frac{1}{2^{\left\lfloor Rn\right\rfloor (1-r)}}\right)}{\left\lfloor r\left\lfloor Rn\right\rfloor \right\rfloor }.\label{eq:Proof3}
\end{equation}
Choosing $\sigma_{n}=\Phi_{n}[\rho_{n}]$ with some sequence of states
$\rho_{n}$ and LOCC protocols $\Phi_{n}$ we obtain
\begin{align}
F\left(\ket{\psi^{-}}\!\bra{\psi^{-}}^{\otimes\left\lfloor Rn\right\rfloor },\Phi_{n}[\rho_{n}]\right) & \leq\frac{E_{\mathrm{r}}(\Phi_{n}[\rho_{n}])+\log_{2}\left(1+\frac{1}{2^{\left\lfloor Rn\right\rfloor (1-r)}}\right)}{\left\lfloor r\left\lfloor Rn\right\rfloor \right\rfloor }\label{eq:SingletsRN}\\
 & \leq\frac{E_{\mathrm{r}}(\rho_{n})+\log_{2}\left(1+\frac{1}{2^{\left\lfloor Rn\right\rfloor (1-r)}}\right)}{\left\lfloor r\left\lfloor Rn\right\rfloor \right\rfloor },\nonumber 
\end{align}
where we used the fact that the relative entropy of entanglement does
not increase under LOCC~\citep{PhysRevLett.78.2275}. 

In the next step, recall that it is possible to distill the states
$\ket{\psi}$ into singlets at rate $S(\psi^{A})$. In more detail,
for any $R<S(\psi^{A})$ there exists some $\alpha>0$ and a sequence
of LOCC protocols $\Phi_{n}$ such that the following inequality holds
for all $n$ large enough~\citep{Hayashi_2003}:
\begin{equation}
F\left(\Phi_{n}\left[\psi^{\otimes n}\right],\ket{\psi^{-}}\!\bra{\psi^{-}}^{\otimes\left\lfloor Rn\right\rfloor }\right)\geq1-2^{-\alpha n},\label{eq:Distillation-1}
\end{equation}
we also refer to Proposition~\ref{prop:PureExponential} for more
details. From Eq.~(\ref{eq:Distillation-1}) it follows that 
\begin{equation}
D_{B}\left(\Phi_{n}\left[\psi^{\otimes n}\right],\ket{\psi^{-}}\!\bra{\psi^{-}}^{\otimes\left\lfloor Rn\right\rfloor }\right)\leq\sqrt{2-2\sqrt{1-2^{-\alpha n}}},
\end{equation}
where $D_{B}$ is the Bures distance defined in Section~\ref{sec:Preliminaries}.
Recalling that the Bures distance fulfills the triangle inequality
and the data processing inequality, we find for any sequence of states
$\rho_{n}$:
\begin{align}
D_{B}\left(\ket{\psi^{-}}\!\bra{\psi^{-}}^{\otimes\left\lfloor Rn\right\rfloor },\Phi_{n}\left[\rho_{n}\right]\right) & \leq D_{B}\left(\Phi_{n}\left[\psi^{\otimes n}\right],\Phi_{n}\left[\rho_{n}\right]\right)+D_{B}\left(\Phi_{n}\left[\psi^{\otimes n}\right],\ket{\psi^{-}}\!\bra{\psi^{-}}^{\otimes\left\lfloor Rn\right\rfloor }\right)\nonumber \\
 & \leq D_{B}\left(\Phi_{n}\left[\psi^{\otimes n}\right],\Phi_{n}\left[\rho_{n}\right]\right)+\sqrt{2-2\sqrt{1-2^{-\alpha n}}}\nonumber \\
 & \leq D_{B}\left(\psi^{\otimes n},\rho_{n}\right)+\sqrt{2-2\sqrt{1-2^{-\alpha n}}},
\end{align}
which is equivalent to 
\begin{equation}
D_{B}\left(\psi^{\otimes n},\rho_{n}\right)\geq D_{B}\left(\ket{\psi^{-}}\!\bra{\psi^{-}}^{\otimes\left\lfloor Rn\right\rfloor },\Phi_{n}\left[\rho_{n}\right]\right)-\sqrt{2-2\sqrt{1-2^{-\alpha n}}}.
\end{equation}
Expressing this in terms of fidelity we, see that for any $R<S(\psi^{A})$
there exists some $\alpha>0$ such that for any sequence of states
$\rho_{n}$ the following holds for all $n$ large enough:
\begin{align}
F\left(\psi^{\otimes n},\rho_{n}\right) & =\left[1-\frac{1}{2}D_{B}^{2}\left(\psi^{\otimes n},\rho_{n}\right)\right]^{2}\\
 & \leq\left[1-\frac{1}{2}\left(D_{B}\left(\ket{\psi^{-}}\!\bra{\psi^{-}}^{\otimes\left\lfloor Rn\right\rfloor },\Phi_{n}\left[\rho_{n}\right]\right)-\sqrt{2-2\sqrt{1-2^{-\alpha n}}}\right)^{2}\right]^{2}\nonumber \\
 & =\left[1-\frac{1}{2}\left(\sqrt{2-2\sqrt{F\left(\ket{\psi^{-}}\!\bra{\psi^{-}}^{\otimes\left\lfloor Rn\right\rfloor },\Phi_{n}\left[\rho_{n}\right]\right)}}-\sqrt{2-2\sqrt{1-2^{-\alpha n}}}\right)^{2}\right]^{2}.\nonumber 
\end{align}
Here, $\Phi_{n}$ is the sequence of LOCC protocols which distills
$\ket{\psi}$ into singlets, see Eq.~(\ref{eq:Distillation-1}).

Using Eq.~(\ref{eq:Er}), we see that for any $R<S(\psi^{A})$ there
is some $\alpha>0$ such that the following holds for all $n$ large
enough:
\begin{align}
E_{\mathrm{r}}\left(\Lambda_{n}^{\psi,\phi}[\rho_{n}]\right) & =\left\lfloor qn\right\rfloor S(\phi^{A})F\left(\psi^{\otimes n},\rho_{n}\right)\label{eq:PsiProof1}\\
 & \leq\left\lfloor qn\right\rfloor S(\phi^{A})\left[1-\frac{1}{2}\left(\sqrt{2-2\sqrt{F\left(\ket{\psi^{-}}\!\bra{\psi^{-}}^{\otimes\left\lfloor Rn\right\rfloor },\Phi_{n}\left[\rho_{n}\right]\right)}}-\sqrt{2-2\sqrt{1-2^{-\alpha n}}}\right)^{2}\right]^{2}\nonumber \\
 & =\left\lfloor qn\right\rfloor S(\phi^{A})\left[1-\frac{1}{2}\left(V^{2}+W^{2}-2VW\right)\right]^{2},\nonumber 
\end{align}
where we have defined 
\begin{align}
V & =\sqrt{2-2\sqrt{F\left(\ket{\psi^{-}}\!\bra{\psi^{-}}^{\otimes\left\lfloor Rn\right\rfloor },\Phi_{n}\left[\rho_{n}\right]\right)}},\\
W & =\sqrt{2-2\sqrt{1-2^{-\alpha n}}}.
\end{align}
By further defining 
\begin{align}
X & =\sqrt{F\left(\ket{\psi^{-}}\!\bra{\psi^{-}}^{\otimes\left\lfloor Rn\right\rfloor },\Phi_{n}\left[\rho_{n}\right]\right)},\\
Y & =\frac{1}{2}\left(W^{2}-2VW\right),\\
Z & =Y^{2}-2XY\label{eq:Z}
\end{align}
we can express Eq.~(\ref{eq:PsiProof1}) as follows: 
\begin{align}
E_{\mathrm{r}}\left(\Lambda_{n}^{\psi,\phi}[\rho_{n}]\right) & \leq\left\lfloor qn\right\rfloor S(\phi^{A})\left[\sqrt{F\left(\ket{\psi^{-}}\!\bra{\psi^{-}}^{\otimes\left\lfloor Rn\right\rfloor },\Phi_{n}\left[\rho_{n}\right]\right)}-\frac{1}{2}\left(W^{2}-2VW\right)\right]^{2}\nonumber \\
 & =\left\lfloor qn\right\rfloor S(\phi^{A})\left[X-Y\right]^{2}\nonumber \\
 & =\left\lfloor qn\right\rfloor S(\phi^{A})\left[X^{2}+Y^{2}-2XY\right]\nonumber \\
 & =\left\lfloor qn\right\rfloor S(\phi^{A})\left[F\left(\ket{\psi^{-}}\!\bra{\psi^{-}}^{\otimes\left\lfloor Rn\right\rfloor },\Phi_{n}\left[\rho_{n}\right]\right)+Z\right].
\end{align}
Using Eq.~(\ref{eq:SingletsRN}) we conclude that for any $r>0$,
$R<S(\psi^{A})$ the following holds for all $n$ large enough: 
\begin{equation}
E_{\mathrm{r}}\left(\Lambda_{n}^{\psi,\phi}[\rho_{n}]\right)\leq\left\lfloor qn\right\rfloor S(\phi^{A})\frac{E_{\mathrm{r}}(\rho_{n})+\log_{2}\left(1+\frac{1}{2^{\left\lfloor Rn\right\rfloor (1-r)}}\right)}{\left\lfloor r\left\lfloor Rn\right\rfloor \right\rfloor }+\left\lfloor qn\right\rfloor S(\phi^{A})Z.
\end{equation}

Note that for any given value of $q$ in the range $0<q<S(\psi^{A})/S(\phi^{A})$
there exist some values for $r$ and $R$ in the range $0<r<1$ and
$0<R<S(\psi^{A})$ such that 
\begin{equation}
\frac{\left\lfloor qn\right\rfloor S(\phi^{A})}{\left\lfloor r\left\lfloor Rn\right\rfloor \right\rfloor }\leq1
\end{equation}
for all $n$ large enough. For any such choice of $q$, $r$ and $R$
we see that the following holds for all $n$ large enough:
\begin{equation}
E_{\mathrm{r}}\left(\Lambda_{n}^{\psi,\phi}[\rho_{n}]\right)\leq E_{\mathrm{r}}(\rho_{n})+\log_{2}\left(1+\frac{1}{2^{\left\lfloor Rn\right\rfloor (1-r)}}\right)+\left\lfloor qn\right\rfloor S(\phi^{A})Z.
\end{equation}

In the final part of the proof we will analyze closer the term $\left\lfloor qn\right\rfloor S(\phi^{A})Z$.
From Eq.~(\ref{eq:Z}) we see that $Z$ can be written in the form
\begin{align}
Z & =W\left[\frac{1}{4}\left(W^{3}-4VW^{2}+4V^{2}W\right)-XW+2XV\right].
\end{align}
Note that each of the terms $V$, $W$, and $X$ is bounded for all
$n$, and moreover $\lim_{n\rightarrow\infty}nW=0$ for all $\alpha>0$,
which implies that $\lim_{n\rightarrow\infty}nZ=0$. From this it
is clear that the term $\left\lfloor qn\right\rfloor S(\phi^{A})Z$
can be made arbitrarily small by choosing large enough $n$. For any
$r<1$, $R>0$ it is further clear that 
\begin{equation}
\lim_{n\rightarrow\infty}\log_{2}\left(1+\frac{1}{2^{\left\lfloor Rn\right\rfloor (1-r)}}\right)=0,
\end{equation}
and the proof is complete.
\end{proof}

\section{Conclusions}

To conclude, we have established a framework for multipartite entanglement
theory that presents a comprehensive solution for asymptotic transformation
rates across all multipartite pure states. Our approach uniquely incorporates
a subtle relaxation to the LOCC paradigm, permitting all transformations
on a multipartite system that can increase the bipartite relative
entropy of entanglement of any state by $\varepsilon$, requiring
that $\varepsilon$ can be made arbitrarily small. The primary finding
of our research is that transformation rates for all multipartite
pure states are fundamentally determined by the bipartite entanglement
entropies of the involved quantum states. These results underscore
the centrality of entanglement entropy in governing quantum state
transitions. 

In the context of a tripartite system, our methodology bridges a crucial
gap between tripartite and bipartite entanglement theory. Although
transformations of multipartite pure states are typically irreversible,
even in our setup, we demonstrate that reversibility can be achieved
in certain meaningful scenarios. This extends to the reversible conversion
between GHZ and $W$ states in a three-qubit setting, as well as the
reversible conversion between a pair of GHZ states and three singlets,
each singlet being shared by a different pair of parties. Additionally,
we establish that in the framework proposed herein, singlets can act
as a reversible entanglement generating set for all tripartite pure
states.

It is worth noting that the results presented in this article conclusively
address an open question posed in~\citep{PhysRevA.63.012307} over
two decades ago. The authors of~\citep{PhysRevA.63.012307} asked
whether there exists a notion of asymptotic state transformations
that would enable reversible interconversion among all multipartite
pure states which possess identical entanglement entropies across
all bipartitions. Our work provides a positive response to this question,
significantly advancing our understanding of entanglement theory in
the multipartite setting.

The results described in this article lead to several intriguing questions
and promising trajectories for further investigation. One such line
of research pertains to the role of catalysis within the AEN framework.
Recent studies indicate that within the LOCC paradigm, asymptotic
transformations bear a close relationship with single-copy transformations
that involve catalysis~\citep{PhysRevLett.127.150503,PhysRevLett.127.080502,PhysRevLett.129.120506,ganardi2023catalytic,lami2023catalysis,datta2022catalysis}.
Currently, it remains an open question whether this relationship extends
to the AEN setting. 

This work was supported by the National Science Centre Poland (Grant No. 2022/46/E/ST2/00115).

\emph{Note added:} After the completion of this manuscript, two independent
works~\citep{hayashi2024,lami2024} have presented a proof of the
generalized quantum Stein's lemma. These works also establish that
the resource theory of entanglement becomes fully reversible in the
bipartite setting when LOCC is extended to asymptotically non-entangling
operations~\citep{Brandao2008,Brandao2010,Brandao2010b}. Given that
asymptotically non-entangling operations form a subset of AEN, this
result further demonstrates that AEN operations also lead to a fully
reversible entanglement theory in the bipartite case.

\appendix
\bibliography{literature}

\end{document}